\documentclass[11pt,twoside,onecolumn]{IEEEtran}
\usepackage{amsfonts,amsthm,amsmath}
\usepackage{graphicx,color}
\usepackage{listings}
\usepackage{hyperref}

\DeclareMathOperator*{\argmin}{arg\,min}

\newtheorem{theorem}{Theorem}

\begin{document}

\title{On Bandlimited Spatiotemporal Field Sampling with Location and Time
Unaware Mobile Senors}

\author{Sudeep Salgia and Animesh Kumar}

\maketitle

\begin{abstract}

Sampling of physical fields has been a topic that been studied extensively in
literature but has been restricted to a small class of fields like temperature
or pollution which are essentially modelled by the standard second order partial
differential equation for diffusion. Furthermore, a large number of sampling
techniques have been studied from sensor networks to mobile sampling under a
variety of conditions like known and unknown locations of the sensors or
sampling locations and with samples affected by measurement and/or quantization
noise. Also, certain works have also addressed time varying fields incorporating
the difference in known timestamps of the obtained signals. 

It would be of great interest to explore fields which are modelled by more
general constant coefficient linear partial differential equations to address a
larger class of fields that have a more complex evolution. Additionally, this
works address an extremely general and challenging problem, in which such a
field is sampled using an inexpensive mobile sensor such that both, the
locations of the samples and the timestamps of the samples are {\em unknown}.
Moreover, the locations and timestamps of the samples are assumed to be
realizations of two independent {\em unknown} renewal processes. Furthermore,
the samples have been corrupted by measurement noise. In such a challenging
setup, the mean squared error between the original and the estimated signal is
shown to decreasing as $O(1/n)$, where $n$ is the average sampling density of
the mobile sensor.
\end{abstract}


\section{Introduction} % (fold)
\label{sec:introduction}
Sampling of smooth spatiotemporally varying fields is a problem that has been
addressed in literature for multiple reasons. Often the aim has been to estimate
the sources in a diffusion field while some papers have addressed the problem of
estimating the field from the samples. Classical approaches towards this study
have generally involved samples from distributed sensor networks corrupted with
measurement noise and often assume that the time instants of the measurements
are precisely known or have a control over them. Sampling using a mobile sensor
has been a problem that has received attention of late. This problem also has
been well studied in literature for temporally fixed fields with samples from
precisely known sampling locations to unknown locations. Also, certain works do
consider a time variation of field and have addressed generally using known
sampling locations and known time stamps.

Extensive research has been done in sampling fields which can be modelled using
the diffusion equation. In fact, almost all models studying spatiotemporal
fields assume a model of the field which is evolving according to the standard
diffusion equation. However, to the best knowledge of the authors, there has
been no work in sampling fields governed by any linear partial differential
equation (PDE) with constant coefficients. It is important to note here that,
all PDEs are not good models for physical fields. Thus, PDEs that are under
consideration here are the ones which can possibly model a physical field. An
important criteria here is that the energy has to be finite and typically
decreasing with time due to finite support considerations and inherent
"diffusive" nature. This be will quantified in a later section of this work.
Furthermore, sampling of fields varying with time has generally been studied in
specific, constrained environments like uniform sampling, (in space or time or
both), or non uniform sampling with precisely known spatial locations and time
stamps, or unknown locations of either a very slowly varying fields or at known
time instants. The primary motivation of this work is to analyze sampling in a
highly generalized setup of any physical field. Such scenarios, are rather
common in real world. An example is sampling of a pollution field using an
inexpensive device. Often adding precision of knowledge of location and time to
sampling system leads to a considerable increase in cost and hence often
inexpensive devices is used which can record the location or the time stamps of
the samples. Modelling realistic scenarios of sampling fields without the
knowledge of spatial locations or time instants is another important aspect of
motivation behind this work.

This work considers a very general model of a smooth field with finite support
that is evolving according a known, linear partial differential equation with
constant coefficients. For mathematical tractability, field is assumed to be one
dimensional but is evolving with time according to the known PDE. The smoothness
of the field is modelled by is spatial bandlimitedness. However, the field need
not be bandlimited in time, which in fact, is often the case. It is important to
note here that if a field and its certain number of temporal derivatives are
known to be spatially bandlimited at $t = 0$, then it can be concluded that the
field will be always be bandlimited if it evolves according to the given PDE.
The number of the temporal derivatives depends on the degree of the time
derivative in the PDE and result has been proven in this work (Appendix A).
Also, it is considered that location and time stamps of {\em all} samples are
{\em unknown} and are realizations of two independent {\em unknown} renewal
processes. Also, the samples are assumed to be corrupted with measurement noise
which is independent of all other processes and is assumed to have zero mean and
a finite variance. There are no other assumptions about the nature of the noise
or its statistics. The primary way of reducing error in this setup would be
oversampling, like any other setup involving spatial sampling with unknown
locations. However, it is important to note that the oversampling is only in the
spatial domain and not in time domain. The number of samples will solely be
governed by spatial sampling density.

The main result of this paper is that if a spatially bandlimited field evolving
according to a constant coefficient linear PDE, is sampled such that sampling
locations and instants are {\em unknown} and are obtained from {\em unknown}
renewal processes that are independent, then the mean square error between the
estimated field from the noisy samples and the original field at $t = 0$
decreases as $O(1/n)$, where $n$ is the average sampling density, that is, the
expected number of samples over the support of the field.

\begin{figure}[!htb]
\centering
\includegraphics[width =0.4\textwidth, angle = 0]{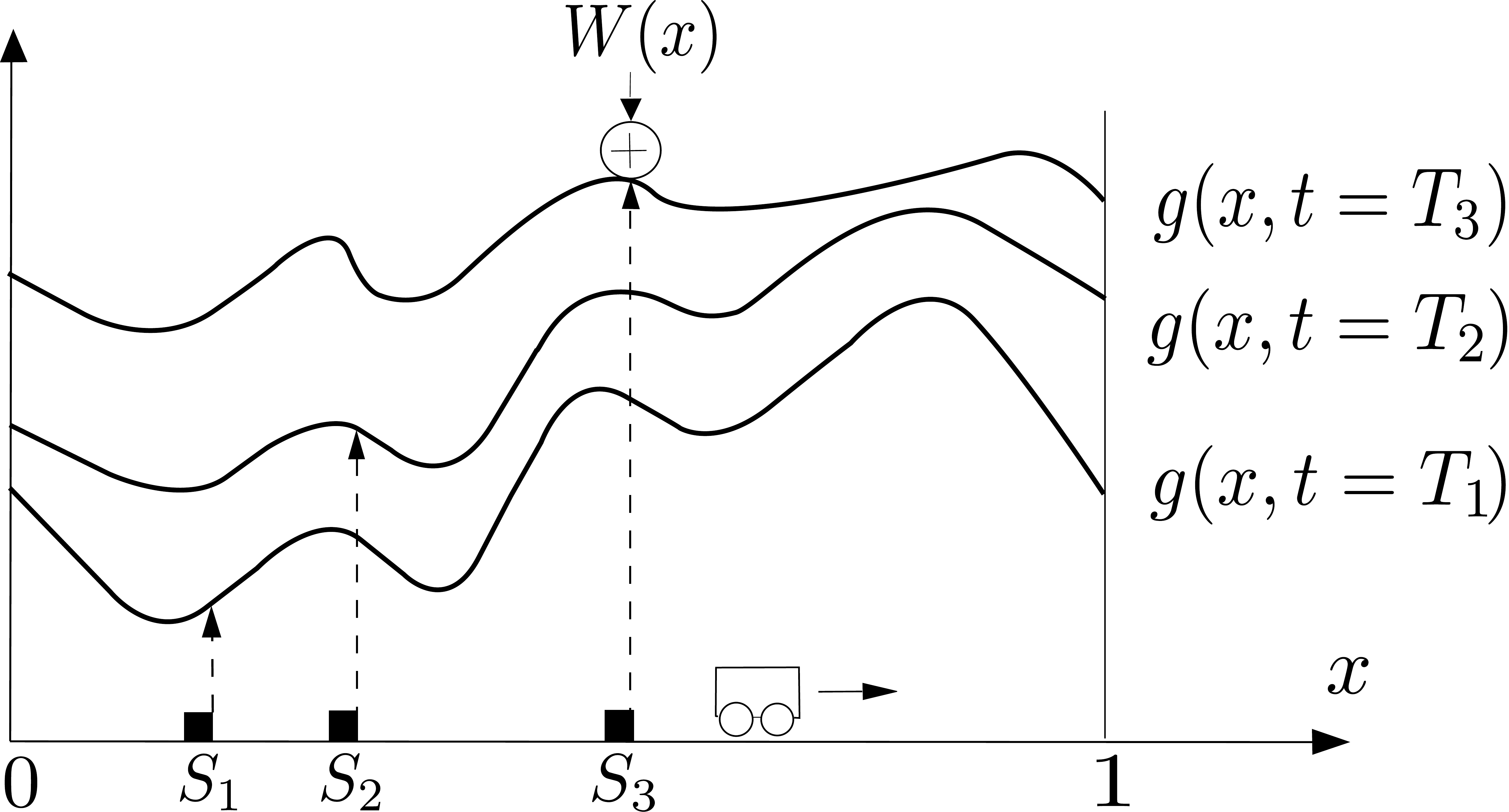}
\caption{\label{Fig:FieldDesciption} The mobile sampling scenario under study
is illustrated. A mobile sensor collects the spatial field's values at unknown
locations denoted bt $S_1, S_2, \dots$. Note how that samples are obtained at
different time instants, which are also unknown. The field is evolving with time
and hence were are getting samples of technically different field equations.
This can be seen in the illustration. It is also assumed that the samples are
affected by additive and independent noise process $W(x)$. Our task is to
estimate $g(x,t)$ from the readings $g(S_1, T_1)+W(s_1), \ldots ,g(S_m,
T_m)+W(s_m)$.}
\end{figure}

 {\em Prior Art}: Estimating a spatiotemporally varying field or sources in a
diffusive field has been a problem addressed in literature. Classically, the
problem involved estimating sources in a diffusive field from distributed sensor
networks, which often requires solving an inverse problem, i.e., inferring
certain characteristics of the field like its distribution at any instant, from
a small number of samples of the field. What makes such problems difficult is
the fact that these inverse problems involving diffusion equation are known to
be severely ill-conditioned \cite{ill_conditioned}. Several works have thus
looked at different forms to regularize the problem. Nehorai~et~al.~\cite{Nehorai} invoked spatial sparsity of sources and studied the detection and
the localization of a single vapor-emitting source using a maximum likelihood
estimator. Two different approaches to the reconstruction of a sparse source
distributions, one involving spatial super-resolution\cite{super_res} and other
on an adaptive spatio-temporal sampling scheme\cite{sparse_sources}, were
introduced by Lu and Vetterli. Another method\cite{localized_sources} using
Prony's method has been proposed. Ranieri~et~al.~\cite{Ranieri} employed
compressed sensing on a discrete grid to estimate the field which has been
extended to real line\cite{Dokmanic}. Apart from the standard diffusion PDE, the
Poisson PDE has also been studied and solutions to that using finite elements
has been also proposed\cite{FEM}, \cite{FEM_1}. The scenario when spatial
sparsity is not realistic has also been studied\cite{aliasing}.  Sampling using
a mobile sensor has been a topic of recent interest\cite{Vett_samp_orig1},
\cite{Vett_samp_orig2}. Estimating fields using mobile sampling has been a
well-studied problem. This problem reduces to the classical sampling and
interpolation problem (as described in \cite{samp_interp_1},
\cite{samp_interp_2}, \cite{samp_interp_3} ), if the samples are collected on
{\em precisely known} locations in absence of any measurement noise. A more
generic version with precisely known locations, in presence of noise, both
measurement and quantization, has also been addressed (refer \cite{known_loc_1}
- \cite{known_loc_6}). Sampling and reconstruction of bandlimited signals from
samples taken at unknown locations has also been studied using a number of
variations of sampling models (see \cite{prior_art_1} - \cite{prior_art_6}).
This work is different from all previous works in the following ways: (i) The
field is considered to be evolving with a known constant linear PDE, which need
not be the diffusion equation. It can be any linear PDE as long as it is a
feasible model for a physical field. (ii) Both the sampling locations and the
timestamps of the samples are unknown, unlike previous works where either fields
are considered to be temporally fixed or time stamps are assumed to be known.

{\em Notation}: The spatiotemporally varying field will be denoted as $g(x,t)$.
The gradient of $g$ is defined as $\nabla g = \left[\frac{\partial}{\partial x}
g(x,t), \frac{\partial}{\partial t} g(x,t)\right]$. $n$ denotes the average
sampling density, while $M$ is the random variable which denotes the number of
samples taken over the support of the field. All vectors will be denoted in
bold. The $\mathcal{L}^{\infty}$ norm of a vector $\mathbf{x}$ will be denoted
by $||\mathbf{x}||_{\infty}$. The expectation operator will be denoted by
$\mathbb{E}[.]$. The expectation is over all the random variables within the
arguments. The trace of a matrix $A$ will be denoted by $\text{tr}(A)$. The set
of natural number, integers, reals and complex numbers will be denoted by
$\mathbb{N}, \mathbb{Z}, \mathbb{R}$ and $\mathbb{C}$ respectively. Also, $j =
\sqrt{-1}$.

\section{Field, Sampling and Noise Model, Distortion Critertia}
\label{sec:field_sampling_and_noise_model_distortion_critertia}

\subsection{Field Model}
\label{sub:field_model}

The field is considered to be spatially smooth over a finite support, one
dimensional in space and evolving with time according to a Partial Differential
Equation. The PDE is linear with constant coefficients and is assumed to be
known and given by
\begin{align}\label{diff_eq}
\sum_{i = 0}^m p_i \frac{\partial^i}{\partial t^i} g(x,t) = \sum_{i = 0}^{m'}
q_i \frac{\partial^i}{\partial x^i} g(x,t) 
\end{align}
where, $\displaystyle \frac{\partial^0}{\partial y^0} f(y) = f(y)$. This will be
represented throughout the paper in the terms of polynomials. Define two
polynomials as 
\begin{align}\label{polynomials}
p(z) = \sum_{i = 0}^m p_i z^i \ ; \ q(z) = \sum_{i = 0}^{m'} q_i z^i
\end{align}
where the coefficients are same as in the differential equation. If for notation
purposes, $\displaystyle \left( \frac{\partial}{\partial z} \right)^l = \left(
\frac{\partial^l}{\partial z^l} \right)$, then the original equation can be
written as 
\begin{align}\label{diff_eq_polynomial}
p\left(\frac{\partial}{\partial t}\right) g(x,t) =
q\left(\frac{\partial}{\partial x}\right) g(x,t)
\end{align}
To incorporate the smoothness of the field, the field is assumed to be
bandlimited. It is important to note we need to ensure that the field is
bandlimited as it evolves with time. Intuitively speaking, this condition
should hold true if we have know the field is spatially bandlimited at $t = 0$,
because time evolution is unlikely to affect the spatial bandlimitedness.
Formally speaking, if the degree of the polynomial $p$ is $m$, then if $m - 1$
partial derivatives of $g(x,t)$ along with $g(x,t)$ are spatially bandlimited
then the function $g(x,t)$ will be spatially bandlimited $\forall t \geq 0$ if
the field evolves according to the given PDE. A detailed proof of this has been
given in Appendix A. We will assume that such a condition holds for the field
under consideration. This helps us ensure that field will always be bandlimited.
Since the field considered has a finite support, assumed WLOG as $[0,1]$, it can
be represented as
\begin{align}\label{gxt}
g(x,t) = \sum_{k = -b}^b a_k(t) \exp (j 2 \pi k x) \ ; \ a_k(t) = \int_{0}^{1}
g(x,t) \exp(-j 2\pi k x) dx
\end{align}
%
% subsection field_model (end)
%
Also the field is assumed to be bounded. That is, $|g(x,t)| \leq 1 \ \forall
(x,t)$. 

\subsection{Distortion Criteria}
\label{sub:distortion_criteria}

To measure the distortion, we will use a simple mean squared error between the
estimated field and the actual field. All measurements will be considered for $t
= 0$ i.e., the mean square error will be considered between the estimated field
at $t = 0$ and the actual field at $t = 0$. Let the estimated field be
$\hat{G}(x,t)$ and its Fourier coefficients be $\hat{A}[k]$, then the distortion
criterion is defined as 
\begin{align}
\mathcal{D}\left[\hat{G},g\right] & = \mathbb{E}\left[\int_0^1 |\hat{G}(x,t) -
g(x,t)|^2\right] \bigg|_{t=0} \\
& =  \mathbb{E}\left[\sum_{k = -b}^b |\hat{A}_k(t) - a_k(t)|^2 \right]
\bigg|_{t=0} \\
& =  \sum_{k = -b}^b \mathbb{E}\left[ |\hat{A}_k(0) - a_k(0)|^2 \right] 
\end{align}
%

% subsection distortion_criteria (end)

\subsection{Sampling Model} % (fold)
\label{sub:sampling_model}

The sampling model is in this case a renewal process based sampling model,
similar to the one in \cite{unknown_loc}. Let $X_1, X_2, \dots $ denote the
intersample distances and $N_1, N_2, \dots $ denote the intersample time
intervals. The sampling model employed assumes the spatial and the temporal
separations as realizations of two independent renewal processes. In other
words, $X_1, X_2, \dots$ are i.i.d.~random variables having a common
distribution $X >0$ and $N_1, N_2, \dots$ are also i.i.d.~random variables
having a single common distribution $N > 0$, such that $X_i$ and $N_j$ are
independent random variables for all values of $i, j \in \mathbb{N}$. Using
these intersample distances, the sampling locations, $S_n$, are given by
$\displaystyle S_n = \sum_{i =1}^n X_i$. The sampling is done over the interval
$[0,1]$, the support of the field, and $M$ is the random number of samples that
lie in the interval i.e. it is defined such that, $S_M \leq 1$ and $S_{M+1} >
1$. Thus $M$ is a well defined measurable random variable\cite{meas_var}. Note
that, the number of samples in the interval only depend on the spatial sampling
density and not on the temporal counterpart.  

For the purpose of ease of analysis and tractability, the support of the
distributions of $X$ and $N$ are considered to be finite and inversely
proportional to the sampling density. Hence, it is assumed that 
\begin{align}
0 < X \leq \frac{\lambda}{n}, 0 < N \leq \frac{\mu}{n} \text{  and  }
\mathbb{E}[X] = \mathbb{E}[N] = \frac{1}{n},
\end{align}
where $\lambda, \mu > 1$ are parameters that characterize the support of the
distributions. Both are finite numbers, independent of the average sampling
density $n$ and much smaller than $n$ i.e., $\lambda, \mu \ll n$. These would be
important factors that govern the constant of proportionality in the expected
error of the estimate. Furthermore, $\lambda$  is an important factor also to
determine the threshold on the minimum number of samples. Applying Wald's
identity\cite{meas_var}, on $S_{M+1}$,
\begin{align}\label{bound_on_EM}
\mathbb{E}[M + 1]\mathbb{E}[X] & = \mathbb{E}[S_{M+1}] \\
 (\mathbb{E}[M] +  1) \frac{1}{n} & = \mathbb{E}[S_{M+1}] \nonumber \\
 \mathbb{E}[M] &= n\mathbb{E}[S_{M+1}] - 1 
\end{align}
By definition, $S_{M + 1} > 1$ and $S_M \leq 1$. Since $S_{M +1} = S_M +
X_{M+1}$, therefore, $S_{M+1} \leq 1 + X_{M+1} \leq 1 + \frac{\lambda}{n}$. Use
these inequalities with equation (\ref{bound_on_EM}), to obtain,
\begin{align}\label{bounds_on_EM}
n - 1 < \mathbb{E}[M] \leq n + \lambda - 1
\end{align}
Also, the bound on each $X_i \leq \frac{\lambda}{n}$, along with $S_{M + 1} > 1$
gives,
\begin{align}\label{upper_bound_on_M}
(M + 1)\frac{\lambda}{n} > 1 \text{ or } M > \frac{n}{\lambda} - 1
\end{align}
Note that all the bounds on the number of samples, a random variable, is
characterized solely in terms of $\lambda$, the parameter that defines support
for the spatial renewal process. There is no involvement of the temporal renewal
process as expected. The only assumption on the sampling model with respect to
time is that it has been assumed that all samples are collected within some time
$T_0$, which is known. That is to say, $T_M \leq T_0$ and $T_{M+1} > T_0$. Note
that this value is variable. This value is important in general to be known in
sampling scenarios especially in case of time varying fields because if it is
too large then the field has likely decayed to a very small value which can lead
to erroneous results. Thus, the knowledge of this is assumed. Since $nN \leq \mu
\ll n$, we expect that $T_0 \ll M$.

\subsection{Measurement Noise Model}
\label{sub:measurement_noise_model}

It will be assumed that the obtained samples have been corrupted by additive
noise that is independent both of the samples and of both the renewal processes.
For simplicity, the noise is considered to be varying only spatially. That is,
at all time instants, the distribution of the noise remains the same, which is
assumed to be {\em unknown} in this work. Hence, $W(x,t) \equiv W(x) $. Thus,
the samples obtained would be sampled versions of $g(x, t) + W(x, t)$, where
$W(x,t) \equiv W(x)$ is the noise. Also, since the measurement noise is
independent, that is for any set of measurements at distinct points $s_1, s_2,
s_3, \dots s_n$, the samples $W(s_1), W(s_2), W(s_3), \dots W(s_n)$ would be
independent and identically distributed random variables. Note that the sampling
instants have not been considered because of the distribution being temporally
static. The only statistics known about the noise  are that the noise is zero
mean and has a finite variance, $\sigma^2$.

\section{Field Estimation from the Obtained Samples}
\label{sec:field_estimation_from_samples}

This section will mainly deal with the estimation of the field from the samples
whose locations and time stamps come from two {\em unknown} independent renewal
processes. Before that, it is essential to analyse the development of the field
under the differential equation. Using the fact that Fourier series are linear
in coefficients, and combining the equations \eqref{gxt} and \eqref{diff_eq},
and using the orthogonality of Fourier basis, we can write,

\begin{flalign}
 && \sum_{i = 0}^m p_i \frac{\partial^i}{\partial t^i} \left( \sum_{k = -b}^b
a_k(t) \exp (j 2 \pi k x) \right) & = \sum_{i = 0}^n q_i
\frac{\partial^i}{\partial x^i} \left(  \sum_{k = -b}^b a_k(t) \exp (j 2 \pi k
x) \right) &&  \nonumber \\
 &&  \sum_{k = -b}^b \left( \sum_{i = 0}^m p_i \frac{\partial^i a_k(t)}{\partial
t^i} \right) \exp (j 2 \pi k x) & = \sum_{k = -b}^b a_k(t) \left( \sum_{i = 0}^n
q_i (j 2\pi k )^i \right) \exp (j 2 \pi k x) && \nonumber \\
 && \sum_{k = -b}^b \left( \sum_{i = 0}^m p_i \frac{\partial^i a_k(t)}{\partial
t^i} \right) \exp (j 2 \pi k x) & = \sum_{k = -b}^b a_k(t) q(j 2 \pi k) \exp (j
2 \pi k x) && \nonumber \\
 &&  \sum_{i = 0}^m p_i \frac{\partial^i a_k(t)}{\partial t^i} - q(j 2 \pi k)
a_k(t) & = 0 \  \ \forall k = -b, \dots,b &&
\end{flalign}
where (a) follows from \ref{polynomials} and (b) uses the orthogonality proprety
for the Fourier basis. This gives us a differential equation for each $a_k(t)$.
To solve for $a_k(t)$, the general method is adopted and the solution is assumed
to be of the form $e^{rt}$. For each $k$, this leads to the polynomial equation, 
\begin{align}\label{r_poly}
  \sum_{i = 0}^m p_i \frac{\partial^i Ae^{rt}}{\partial t^i}  - q(j 2 \pi k)
Ae^{rt} = 0 \\ 
 \left(\sum_{i = 0}^m p_i r^i  - q(j 2 \pi k) \right) Ae^{rt} = 0 \\
   p(r) - q(j 2 \pi k) = 0 
\end{align} 
The solution for $a_k(t)$ is a of the form $Ae^{rt}$, where $r$ is the root of
the above polynomial and $A$ is a constant independent of $t$. Let the roots of
the above polynomial be $r_1(k), r_2(k), \dots r_m(k)$. Note that the roots of
the polynomial are indexed by $k$ as well, implying there is a set of $m$ roots
for each value of $k$. It is essential here to realise that if the field is a
physically feasible one, then $\Re(r_i(k)) \leq 0 \ ; i = 1,2,\dots, m ; k = -b,
\dots,-1, 0, 1 \dots,b $, that is all roots have a non positive real part.
Generally for all physical fields it has to be strictly less than $0$, but we
are allowing the possibility of sustained oscillating (in time) fields.
Furthermore for simplicity of analysis, all of the roots $r_1(k), r_2(k), \dots
r_m(k)$ are considered to be distinct for a given $k$.\footnote{If there is a
repeated root $r, $then the solution will be of the form $e^{rt}$ and also
$te^{rt}$, which will make the problem very complicated. Such cases can also be
treated in a similar manner that has been described in the paper. To be very
specific, if the repeated root is $0$, it can be easily taken into the given
framework by combining all the repeated terms with it. This is because if $r =
0$, then $te^{rt} = t$, which diverges and hence cannot be the solution for a
physical field. Such nuances have been omitted to simplify the description of
the process.} However, it is possible that $r_i(k_1) = r_j(k_2)$ for some $i, j,
k_1 \neq k_2$. This assumption is realistic enough as generally for physical
fields, $m$ is generally very small thus the chance of repeated roots is lesser.
The condition is similar to the one obtained in control theory, where we want
the poles of the closed loop system to lie in the left half plane. Thus, we can
use criteria like the Routh-Hurwitz condition, to ensure the roots have negative
real parts in our case.

In fact, the solution for $a_k(t)$ can thus be written as a linear combination
of these roots. Thus $a_k(t) = \sum_{i = 1}^m a_{ki}(0)\exp(r_i(k) t)$. The
coefficients have been represented so to maintain consistency of representation
of $a_k(t)$ as a function of time. Also $a_{ki}(0)$ are finite constants
independent of everything else. Let $\alpha_k = \max_{i} |a_{ki}(0)|$.
\begin{align}\label{partial_deriv_time_bound}
\bigg|\frac{\partial}{\partial t}g(x,t)\bigg| & = \bigg|\frac{\partial}{\partial
t}\sum_{i = 1}^m a_{ki}(0)\exp(r_i(k) t)\exp (j 2 \pi k x)\bigg| \\
& = \bigg|\sum_{i = 1}^m a_{ki}(0)r_i(k)\exp(r_i(k) t)\exp (j 2 \pi k x)\bigg|
\\
& \leq \sum_{i = 1}^m \bigg| a_{ki}(0)r_i(k)\exp(r_i(k) t)\exp (j 2 \pi k
x)\bigg| \\
& \leq \sum_{i = 1}^m  | a_{ki}(0)| |r_i(k)| \\
& \leq m \alpha_k R 
\end{align}
\begin{align}\label{spatial_deriv_bound}
\bigg|\frac{\partial}{\partial x}g(x,t)\bigg| & = \bigg|\frac{\partial}{\partial
x} \sum_{k = -b}^b a_k(t) \exp (j 2 \pi k x) \bigg| \\
& = \bigg| \sum_{k = -b}^b a_k(t) j 2 \pi k \exp (j 2 \pi k x) \bigg| \\
& \leq \sum_{i = 1}^m \bigg| a_{ki}(0)j 2 \pi k \exp(r_i(k) t)\exp (j 2 \pi k
x)\bigg| \\
& \leq \sum_{i = 1}^m 2 b \pi | a_{ki}(0)|  \\
& \leq m \alpha_k 2b \pi
\end{align}
%
% \begin{equation}\label{spatial_deriv_bound} \bigg|\frac{\partial}{\partial
% x}g(x,t)\bigg| \leq \bigg|\frac{\partial}{\partial x} \sum_{k = -b}^b a_k(t) j
% 2 \pi k \exp (j 2 \pi k x) \bigg| \leq 2b\pi \bigg|\sum_{k = -b}^b a_k(t) \exp
% (j 2 \pi k x) \bigg| \leq 2b\pi \end{equation}
%
The third step follows from triangle inequality, fourth step uses the fact that
$\Re(r_i(k)) \leq 0$ and the bound on $r_i(k)$ uses the Rouche's theorem.  The
value of $R$ can be expressed in terms of $|p_i|$'s, which are finite and so is
the upper bound.  Using the above expression for $a_k(t)$ and using it in
equation \eqref{gxt} to obtain the value at $(x,t) = (S_n, T_n)$, we get, 
\begin{align}\label{gsntn}
g(S_n,T_n) = \sum_{k = -b}^b \sum_{i = 1}^m a_{ki}(0)\exp(r_i(k) T_n) \exp (j 2
\pi k S_n)
\end{align}
The above equation can be written in a vector notation form. Let $e_{k,i}(x,t) =
\exp(r_i(k) t + j 2 \pi k x)$. Define,
\begin{align}\label{def_a_e}
\mathbf{e}_{k,i}(x,t) := [e_{k,1}(x,t),\ e_{k,2}(x,t), \dots, e_{k,m}(x,t)]
\nonumber\\
\mathbf{a}_k := [a_{k1}(0),\ a_{k2}(0),\ a_{k3}(0),\ \dots,\ a_{km}(0)]
\nonumber\\
\mathbf{a} = [\mathbf{a}_{-b}, \dots \mathbf{a}_{-1}, \mathbf{a}_{0},
\mathbf{a}_{1}, \dots \mathbf{a}_{b}]^T \nonumber\\
\mathbf{e}(x,t) = [\mathbf{e}_{-b}(x,t), \dots, \mathbf{e}_{b}(x,t)]^H 
\end{align}
Observe that $\mathbf{a}$ and $\mathbf{e}(x, t)$ are column vectors, while
$\mathbf{e}_k(x, t)$ and $\mathbf{a}_k$ are row vectors. Since $\Re(r_i(k)) \leq
0$, so $|e_{k,i}(x, t)| \leq 1$. This implies,
\begin{align}\label{bounds_on_ek}
||\mathbf{e}_k(x,t)||^2 = \sum_{i = 1}^m |\exp(r_i(k) t + j 2 \pi k x)|^2 \leq
\sum_{i = 1}^m 1 \leq m \nonumber\\
|| \mathbf{e}(x,t) ||^2 = \sum_{k = -b}^b ||\mathbf{e}_k(x,t)||^2 \leq \sum_{k =
-b}^b m = m(2b + 1)  
\end{align}
Therefore, on using equation \eqref{def_a_e}, equation \eqref{gsntn} can be
rewritten as,
\begin{align}\label{gnstn_new}
g(S_n, T_n) = \mathbf{e}^H(S_n,T_n) \ \mathbf{a}
\end{align}
Recall that the sampling locations ($S_n$) and their respective time stamps
($T_n$) are given by,
\begin{align*} S_1 = X_1, \ S_2 = X_1 + X_2, \dots S_n = \sum_{i =1}^n X_i \\
T_1 = N_1, \ T_2 = N_1 + N_2, \dots T_n = \sum_{i =1}^n N_i
\end{align*}
where $X_1, X_2, \dots X_M$ and $N_1, N_2, \dots, N_M$ are all unknown and both
$X_i$'s and $N_i$'s are drawn from independent distributions. The obtained
samples are value of the field at these locations and instants, that have been
corrupted with noise. The observed values are, thus, $g(S_i, T_i) + W(S_i, T_i)\
i = 1, 2, 3, \dots, M$. Define two vectors
\begin{align}\label{def_gw}
\mathbf{g} = [g_1, g_2, \dots g_M]^T \text{ and } \mathbf{w} = [w_1, w_2, \dots
w_M]^T \\  
\text{ where, }  g_i = g(S_i, T_i) \text{ and }  w_i = W(S_i, T_i) \text{ for }
i = 1, 2, 3, \dots M.
\end{align} 
The motivation behind this is to continue to matrix vector notation and hence
all the samples have been stacked up to form a single vector. The vector that
would be obtained on stacking up the samples would be $\mathbf{g}_s = \mathbf{g}
+ \mathbf{w}$. Combining equation (\ref{def_a_e}) and (\ref{def_gw}), we can
write,
\begin{align}\label{def_Y}
\mathbf{g} = \begin{bmatrix}
\mathbf{e}^H(S_1,T_1) \\ \mathbf{e}^H(S_2,T_2) \\ \vdots \\
\mathbf{e}^H(S_M,T_M) 
\end{bmatrix} \mathbf{a} = Y \mathbf{a} \ ; \text{ where } Y = \begin{bmatrix}
\mathbf{e}^H(S_1,T_1) \\ \mathbf{e}^H(S_2,T_2) \\ \vdots \\
\mathbf{e}^H(S_M,T_M) 
\end{bmatrix}
\end{align}
The main idea behind the reconstruction of the field would be that the sampling
location and time instants are ``near'' to the locations and time instants, had
we sampled uniformly both in time and space for $M$ points. Thus, to incorporate
the same into the formulation, let 
\begin{align}
 s_i = \frac{i}{M}, \text{ and } t_i = \frac{iT_0}{M} \text{ for } 1 \leq i \leq
M 
\end{align}
\begin{align}\label{def_Y0}
\mathbf{g}_0 = [g_{u1}, g_{u2}, \dots g_{uM}]^T ; \ 
Y_0 = \begin{bmatrix}
\mathbf{e}^H(s_1,t_1) \\ \mathbf{e}^H(s_2,t_2) \\ \vdots \\
\mathbf{e}^H(s_M,t_M) 
\end{bmatrix},  
\end{align}
where $g_{ui} = g\left(s_i,t_i\right), i = 1,2, \dots, M$. This implies,
$\mathbf{g}_0 = Y_0 \mathbf{a}\ $. Note that $Y_0$ has Vandermonde
structure\cite{vandermode}.
Now, since we expect that the sampling locations are "near" to the grid points,
we can estimate the Fourier coefficients by assuming that samples have been
obtained by multiplying the Fourier coefficient vector by $Y_0$ instead of $Y$.
The best estimate of the Fourier coefficients, $\hat{\mathbf{a}}$, thus would be
\begin{align}\label{def_a_hat}
\hat{\mathbf{a}} = \argmin_{\mathbf{b}} = ||\mathbf{g}_s - Y_0 \mathbf{b}||^2
\end{align}
It is important to note here that instead of $\mathbf{g}$, we have used
$\mathbf{g}_s$ since that is the best knowledge we have about $\mathbf {g}$.
Since the main way to achieve to estimate the field relies on oversampling, the
sampling density will be generally very large and thus, $n > m(2b +1 )$, making
this problem a standard least squares estimation problem. The solution to this
problem is well known and uses the pseudoinverse of the matrix. Therefore, 
\begin{align}\label{a_hat}
\hat{\mathbf{a}} = (Y_0^H Y_0)^{-1} Y_0^H \mathbf{g}_s \\
\mathbf{a} = (Y_0^H Y_0)^{-1} Y_0^H \mathbf{g}_0
\end{align}
The second equation is obtained in a similar manner. However, it is important
to realize at this point that the first equation is a least-square {\em
estimate} because of the unknown locations and noise while the second equation
is an {\em exact} solution. Having defined all the above quantities, we can go
ahead and estimate the error using the distortion criteria mentioned in the
above section.
\begin{align}
\sum_{k = -b}^b \mathbb{E}\left[\bigg|\hat{A}_k(0) - a_k(0)|^2\right] & =
\sum_{k = -b}^b \mathbb{E}\left[|\sum_{i = 1}^m \hat{A}_{ki}(0) - \sum_{i = 1}^m
a_{ki}(0)\bigg|^2\right]  \nonumber \\
& =  \sum_{k = -b}^b \mathbb{E}\left[\bigg|\sum_{i = 1}^m \big( \hat{A}_{ki}(0)
- a_{ki}(0) \big)\bigg|^2\right]  \nonumber \\
&  \leq \sum_{k = -b}^b \mathbb{E}\left[m \sum_{i = 1}^m \big| \hat{A}_{ki}(0) -
a_{ki}(0) \big|^2 \right] \nonumber \\
& = m \sum_{k = -b}^b \mathbb{E}\left[\big|| \hat{\mathbf{a}}_k - \mathbf{a}_k
\big||^2 \right]  \nonumber\\
& \leq m \sum_{k = -b}^b \mathbb{E}\left[\big|| \hat{\mathbf{a}} - \mathbf{a}
\big||^2 \right] \nonumber \\
& = m(2b +1) \mathbb{E}\left[\big|| \hat{\mathbf{a}} - \mathbf{a} \big||^2
\right]
\end{align}
This is the estimate for the Fourier coefficients of the field and distortion
criteria expressed in that estimate. We will establish the bound on the
estimation error as the main result in this work.

\begin{theorem}
Let $\hat{\mathbf{a}}$ and $\mathbf{a}$ be as defined in equation \eqref{a_hat}.
Under the sampling model discussed and the corruption by the measurement noise,
the following result holds
\begin{align*}
\mathbb{E} \left[ || \hat{\mathbf{a}} - \mathbf{a} ||^2\right] \leq \frac{C'}{n}
\end{align*}
where $n$ is the average sampling density and $C'$ is a positive constant
independent of $n$. It depends on the bandwidth, $b$ of the signal, the support
parameters of the renewal processes, $\lambda$ and $\mu$, the coefficients of
the PDE and the noise variance, $\sigma^2$. The dependence on $b$, $\lambda$,
$\mu$ and $\sigma^2$ is such that if these constant would increase, the
proportionality constant would increase, worsening the bound. The dependence on
the coefficients of the PDE is in a very non linear way through the roots of the
equations whose almost all coefficients are determined by these values.
Correspondingly, the distortion error can be bounded as $\frac{m(2b +1)C'}{n}$
\end{theorem}
\begin{proof}
Thus, now we will upper bound $\displaystyle \mathbb{E}\left[\big|
\hat{\mathbf{a}} - \mathbf{a} \big|^2 \right]$ to obtain a bound on distortion.
Letting, $A = (Y_0^H Y_0 )^{-1} Y_0^H$, we can write,
\begin{align}\label{main_step}
 \mathbb{E} \left[ || \hat{\mathbf{a}} - \mathbf{a} ||^2\right] & = \mathbb{E}
\left[ || (Y_0^H Y_0 )^{-1} Y_0^H \mathbf{g}_s  - \mathbf{a} ||^2\right] \\
 & = \mathbb{E} \left[ || (Y_0^H Y_0 )^{-1} Y_0^H (\mathbf{g} + \mathbf{w}) -
(Y_0^H Y_0 )^{-1} Y_0^H \mathbf{g}_0 ||^2\right] \\
 & = \mathbb{E} \left[ || A(\mathbf{g} + \mathbf{w}- \mathbf{g}_0)||^2\right] \\
 & \leq 2\mathbb{E} \left[ || A(\mathbf{g} - \mathbf{g}_0)||^2\right]  +
2\mathbb{E} \left[ || A\mathbf{w}||^2\right]\\
 & \leq 2 \mathbb{E} \left[ \lambda_{\max}^A || \mathbf{g} -
\mathbf{g}_0||^2\right] + 2\mathbb{E} \left[ || A\mathbf{w}||^2\right]\\
\end{align}
where second step follows from equation \eqref{a_hat} and definition of
$\mathbf{g}_s$, the fourth step from Cauchy-Scwharz inequality and
$\lambda_{\max}^A$ is the largest eigenvalue of $A^H A$. Both the terms, along
with the bound on $\lambda_{\max}^A$, in the last step will be analyzed
separately to obtain the bound on the error. Now, 
\begin{align}\label{bound_g}
|| \mathbf{g} - \mathbf{g}_0||^2 & = \sum_{i =1}^M |g(S_i,T_i) - g(s_i, t_i)|^2
\nonumber \\
& \leq \sum_{i =1}^M \left(\max_{\mathbf{x},t} ||\nabla g||_2 \right) \big\{
|S_i - s_i|^2 + |T_i - t_i|^2 \big\} \nonumber\\
& \leq C_0 \bigg\{ \sum_{i = 1}^M \bigg|S_i - \frac{i}{M}\bigg|^2 + \sum_{i =
1}^M \bigg|T_i - \frac{iT_0}{M}\bigg|^2 \bigg\} 
\end{align}
The second step follows from the fact that for a smooth function
$h(\mathbf{x})$, for $\mathbf{x} \in \mathbb{R}^n$, $\displaystyle
|h(\mathbf{x}_1) - h(\mathbf{x}_2)| \leq \left( \max_{\mathbf{x}} ||\nabla h||
\right) ||\mathbf{x}_1 - \mathbf{x}_2||$. The third step uses the fact that
$||\nabla g||_2$ is upper bounded. This follows from the fact that $||\nabla
g||_2^2 = \left(\dfrac{\partial}{\partial x}g(x,t)\right)^2 +
\left(\dfrac{\partial}{\partial t}g(x,t)\right)^2$ along with the bounds on
partial derivatives. From \eqref{partial_deriv_time_bound} and
\eqref{spatial_deriv_bound}, we can write, 
\begin{align}
||\nabla g||_2^2 = \left(\dfrac{\partial}{\partial x}g(x,t)\right)^2 +
\left(\dfrac{\partial}{\partial t}g(x,t)\right)^2 \leq (m \alpha_k R)^2 + (m
\alpha_k R2 b \pi)^2 \leq C_0^2
\end{align}
for some $C_0 \in \mathbb{R}$. 
%
% where the first inequality follows from the complex mean value theorem and the
% second one has been proven in Appendix B in detail. The first inequality using
% complex mean value theorem has also been addressed in Appendix B. To analyse
% the behaviour of $\lambda_{\max}$, consider,
%
\begin{align}\label{lambda_bound1}
\lambda_{\max}^A & \overset{(a)}{\leq} \text{tr} (A^H A) \\
& = \text{tr} (AA^H) \\
& = \text{tr} ((Y_0^H Y_0 )^{-1} Y_0^H Y_0 (Y_0^H Y_0 )^{-1}) \\
& = \text{tr} ((Y_0^H Y_0 )^{-1}) \\
\end{align}
%
% %
(a) follows from the fact that trace of a matrix is the sum of its eigenvalues
and since $A^H A$ is symmetric, all its eigenvalues will be non negative
therefore, the sum will be greater than the largest eigenvalue.  For the second
term in the equation \eqref{main_step}, the structure of the noise model can be
exploited to simplify the expression. Note that using the assumptions on the
noise model, $\mathbb{E}[\mathbf{w}] = 0$ and $\mathbb{E}[\mathbf{ww}^T] =
\sigma^2 I$, where $I$ is the identity matrix.
\begin{align}\label{noise_bound}
\mathbb{E} \left[ || A\mathbf{w}||^2\right] & = \mathbb{E} \left[ \mathbf{w}^T
(A^H A \mathbf{w}) \right] \\
& \overset{(a)}{=} \mathbb{E} \left[ \text{tr} \left(\mathbf{w}^T (A^H A
\mathbf{w})\right) \right]  \\
& \overset{(b)}{=} \mathbb{E} \left[ \text{tr} \left((A^H A \mathbf{w})
\mathbf{w}^T \right)  \right] \\
& \overset{(c)}{=} \text{tr} \left( \mathbb{E} \left[ A^H A \mathbf{w}
\mathbf{w}^T  \right] \right)  \\
& \overset{(d)}{=} \text{tr} \left( \mathbb{E} \left[ A^H A \right] \mathbb{E}
\left[ \mathbf{w} \mathbf{w}^T  \right] \right)  \\
& = \text{tr} \left( \mathbb{E} \left[ A^H A \right] \sigma^2 I \right)  \\
& = \mathbb{E} \left[ \text{tr} (A^H A  \sigma^2 I) \right] \\
& = \mathbb{E} \left[ \sigma^2 \text{tr}(A^H A) \right] \\
& = \sigma^2 \mathbb{E} \left[ \text{tr} ((Y_0^H Y_0 )^{-1}) \right] \\
\end{align}
where, (a) uses the fact that $|| A\mathbf{w}||^2$ is scalar hence, it equals
its trace, (b) follows from $\text{tr}(AB) = \text{tr}(BA)$, (c) uses linearity
of expectation and the trace operator, and (d) is a result of independence of
noise and sampling.
Thus, from equation \ref{main_step}, \ref{lambda_bound1} and \ref{noise_bound},
it is clear that characterizing the bound on $\text{tr} ((Y_0^H Y_0 )^{-1})$ is
required and will also suffice for the purpose.

Let $\lambda_1, \lambda_2, \dots, \lambda_{m(2b+1)}$ be eigenvalues of $Y_0^H
Y_0$. Since the matrix is symmetric, $\lambda_i \geq 0 \ i = 1,2, \dots, m(2b +
1)$. Using the property of eigenvalues, the eigenvalues of $(Y_0^H Y_0)^{-1}$,
will be $\dfrac{1}{\lambda_1}, \dfrac{1}{\lambda_2}, \dots,
\dfrac{1}{\lambda_{m(2b+1)}}$. Let $\lambda_{\max}$ and $\lambda_{\min}$ be the
maximum and minimum eigenvalues of $Y_0^H Y_0$. Therefore,
$\dfrac{1}{\lambda_{\min}}$ and $\dfrac{1}{\lambda_{\max}}$ will be the maximum
and minimum eigenvalues of $(Y_0^H Y_0)^{-1}$. Applying the Polya-Szego
inequality~\cite{pg_ineq} on the sequence formed by eigenvalues of $Y_0^H Y_0$
and by those of $(Y_0^H Y_0)^{-1}$, we can write
\begin{align}
\dfrac{\sum_{i = 1}^{m(2b+1)} \lambda_i \ \sum_{i = 1}^{m(2b+1)}
(1/\lambda_i)}{(\sum_{i = 1}^{m(2b+1)} \sqrt{\lambda_i. 1/\lambda_i})^2} \leq
\frac{1}{4} \left(\frac{\lambda_{\max}}{\lambda_{\min}} +
\frac{\lambda_{\min}}{\lambda_{\max}}\right)^2 
\end{align}
Noting that $\displaystyle \text{tr} (Y_0^H Y_0) = \sum_{i = 1}^{m(2b+1)}
\lambda_i, \text{tr} ((Y_0^H Y_0 )^{-1}) = \sum_{i = 1}^{m(2b+1)} (1/\lambda_i)$
and $\kappa$ is the condition number of the matrix $Y_0^H Y_0$, it can be
written as,
\begin{align}\label{tr_y0h_y0_inv}
 \text{tr} (Y_0^H Y_0) \ \text{tr} ((Y_0^H Y_0 )^{-1}) \leq
\frac{m^2(2b+1)^2}{4} \left(\kappa + \frac{1}{\kappa}\right)^2 \nonumber \\
  \text{tr} ((Y_0^H Y_0 )^{-1}) \leq \frac{m^2(2b+1)^2}{4 \ \text{tr} (Y_0^H
Y_0)} \left(\kappa + \frac{1}{\kappa}\right)^2 
\end{align}
Consider, \begin{align}\label{tr_y0h_y0} \text{tr} (Y_0^H Y_0) & = \text{tr }
\left( [\mathbf{e}(s_1, t_1), \mathbf{e}(s_2, t_2), \dots, \mathbf{e}(s_M, t_M)
] \begin{bmatrix}
\mathbf{e}^H(s_1,t_1) \\ \mathbf{e}^H(s_2,t_2) \\ \vdots \\
\mathbf{e}^H(s_M,t_M) 
\end{bmatrix} \ \right) \nonumber \\
& = ||\mathbf{e}(s_1, t_1)||^2 + ||\mathbf{e}(s_2, t_2)||^2 + \dots +
||\mathbf{e}(s_M, t_M)||^2 \nonumber\\
& \overset{(a)}{=} \sum_{i = 1}^M \sum_{k = -b}^b ||\mathbf{e}_k(s_i, t_i)||^2
\nonumber \\
& \overset{(b)}{=} \sum_{i = 1}^M \sum_{k = -b}^b \sum_{j = 1}^m |e_k^j(s_i,
t_i)|^2 \nonumber\\
& = \sum_{i = 1}^M \sum_{k = -b}^b \sum_{j = 1}^m \exp(2\Re ( r_j(k) ) t_i )
\nonumber \\
& = \sum_{j = 1}^m \sum_{k = -b}^b \sum_{i = 1}^M \exp \left( 2\Re ( r_j(k) )
\frac{iT_0}{M}\right) 
\end{align}
Both (a) and (b) follow directly from equation \eqref{def_a_e}. For a given
value of $j$ and $k$, the sum $\displaystyle \sum_{i = 1}^M \exp \left( 2\Re (
r_j(k) ) \frac{iT_0}{M}\right)$ can be considered as a scalar multiple of
Riemann sum approximation of the integral $\displaystyle \int_0^{T_0} \exp (
2\Re ( r_j(k) ) t) dt$ with partitions chosen uniformly over the interval. The
only difference here that the terms in the sum are not multiplied by the
interval difference, which in this case is the same for all intervals and hence
can be taken out as a scalar. It is interesting to note here that the value of
this integral can be used as a bound on the value of the sum because of that
fact that $\exp ( 2\Re ( r_j(k) ) t)$ is decreasing since $\Re ( r_j(k) ) \leq 0
\ \forall j,k$. The bound can be obtained as,
\begin{align}
\int_0^{T_0} \exp ( 2\Re ( r_j(k) ) t) dt & =  \sum_{i = 1}^M
\int_{\frac{(i-1)T_0}{M}}^{\frac{iT_0}{M}} \exp ( 2\Re ( r_j(k) ) t) dt
\nonumber\\
& \overset{(a)}{\leq} \sum_{i = 1}^M \int_{\frac{(i-1)T_0}{M}}^{\frac{iT_0}{M}}
\exp \left( 2\Re ( r_j(k) ) \frac{(i-1)T_0}{M}\right) dt \nonumber \\
& = \sum_{i = 1}^M  \exp \left( 2\Re ( r_j(k) ) \frac{iT_0}{M}\right) \exp
\left( -2\Re ( r_j(k) ) \frac{T_0}{M}\right)
\int_{\frac{(i-1)T_0}{M}}^{\frac{T_0}{M}}  dt \nonumber \\
& = \exp \left( -2\Re ( r_j(k) ) \frac{T_0}{M}\right)\frac{T_0}{M} \sum_{i =
1}^M  \exp \left( 2\Re ( r_j(k) ) \frac{iT_0}{M}\right)\nonumber \\
& \leq \exp ( -2\Re ( r_j(k) ) )\frac{T_0}{M} \sum_{i = 1}^M  \exp \left( 2\Re (
r_j(k) ) \frac{iT_0}{M}\right)
\end{align}
where (a) follows from the decreasing nature of $\exp ( 2\Re ( r_j(k) ) t)$ and
the last step uses that fact that $T_0 \leq M$. This implies,
\begin{align}\label{riemann_lower_bound}
 \exp \left( 2\Re ( r_j(k) ) \frac{iT_0}{M}\right) \geq M \exp ( 2\Re ( r_j(k)))
\int_0^{T_0} \exp ( 2\Re ( r_j(k) ) t) dt = M C_{jk}
\end{align}
where $\displaystyle C_{jk} = \exp ( 2\Re ( r_j(k) ) ) \int_0^{T_0} \exp ( 2\Re
( r_j(k) ) t) dt$ is a finite constant for each $j,k$. Using equation
(\ref{riemann_lower_bound}) in equation (\ref{tr_y0h_y0}), we get 
\begin{align}\label{final_bound_on_tr}
\text{tr} (Y_0^H Y_0) \geq \sum_{j = 1}^m \sum_{k = -b}^b  M C_{jk} = M C_3 
\end{align}
where $C_3$ is a finite deterministic constant given by $C_3 = \sum_{j = 1}^m
\sum_{k = -b}^b C_{jk}$. Clearly, it is independent of $n$. Since $Y_0$ is a
Vandermonde matrix, using the result in \cite{cond_number} for Vandermonde
matrices with complex entries, such that all the entries lie in the unit circle,
we can also say that the condition number of $Y_0$ is independent of average
sampling density (i.e., $n$) and is upper bounded by a finite constant
independent of its dimension. More specifically, each realization is independent
of $M$, the number of rows in the matrix and would be a finite constant that
does not scale with $M$. Hence, $\kappa$, the condition number of $Y_0^H Y_0$ is
also upper bounded by a finite constant, $C_k > 0$, independent of $n$.

Since $\kappa \geq 1 \ \exists \ K > 0$, such that the term $\displaystyle
\left(\kappa + \frac{1}{\kappa} \right)^2 \leq K $. Combining this results with
ones obtained in \eqref{tr_y0h_y0_inv}~and~\eqref{final_bound_on_tr},
\begin{align} % \label{final_bound_on_tr_inv}
\text{tr} ((Y_0^H Y_0 )^{-1}) \leq \frac{m^2(2b+1)^2}{4M C_3} K =
\frac{C_t}{M}\\
\end{align}
From \eqref{upper_bound_on_M}, it is noted that $M > \dfrac{n-
\lambda}{\lambda}$ or $\dfrac{1}{M} < \dfrac{\lambda}{n - \lambda}$. This
implies, $\mathbb{E}\left[\dfrac{1}{M}\right] < \left[\dfrac{\lambda}{n -
\lambda}\right]$. Therefore, 
\begin{align} \label{final_bound_on_tr_inv}
\mathbb{E}\left[ \text{tr} ((Y_0^H Y_0 )^{-1}) \right] \leq \mathbb{E}\left[
\frac{C_t}{M} \right] \leq \frac{C_t \lambda}{n - \lambda}
\end{align}
Substituting the results obtained above in equations \eqref{lambda_bound1} and
\eqref{noise_bound} and combining them with the equations \eqref{main_step} and
\eqref{bound_g}, we get, 
\begin{align}
\mathbb{E} \left[ || \hat{\mathbf{a}} - \mathbf{a} ||^2\right] \leq
2\mathbb{E}\left[ \frac{C_t}{M} \bigg\{ \sum_{i = 1}^M \bigg|S_i -
\frac{i}{M}\bigg|^2 + \sum_{i = 1}^M \bigg|T_i - \frac{i}{M}\bigg|^2 \bigg\}
\right] + \frac{2C_t \lambda}{n - \lambda}
\end{align}
From Appendix~B, it is noted that,
\begin{align*}
\mathbb{E}\left[ \frac{1}{M}\sum_{i = 1}^M \bigg|S_i - \frac{i}{M}\bigg|^2
\right] \leq \frac{C_S}{n} \text{ and } \mathbb{E}\left[ \frac{1}{M}\sum_{i =
1}^M \bigg|T_i - \frac{iT_0}{M}\bigg|^2 \right] \leq \frac{C_T}{n}
\end{align*}.
Therefore, we can conclude that 
\begin{equation}
\mathbb{E} \left[ || \hat{\mathbf{a}} - \mathbf{a} ||^2\right] \leq \frac{C_t
C_S}{n} + \frac{C_t C_T}{n} + \frac{C_t \lambda}{n - \lambda} \leq \frac{C'}{n}
\end{equation}
This completes the proof.
\end{proof}

\section{Simulations} % (fold)
\label{sec:simulations}

This section presents the results of simulations and explanations to the same.
The simulations have been presented in Fig. . The simulations help ease out
verification of the results obtained of different PDEs and analyze the effect of
sampling density.

Firstly, for the purpose of analysis, a field $g(x, t)$ with $b = 3$ is
considered and its Fourier coefficients have been generated using independent
trials of a Uniform distribution over $[-1,1]$ for all real and imaginary parts
separately. It is ensured that the field is real by using conjugate symmetry.
Finally, the field is scaled to have $|g(x)| \leq 1$. Three differential
equations have been considered for the purpose and Fourier coefficients have
been reused. The simulations have been carried out using the following PDEs
\begin{flalign}\label{PDEs}
&&\frac{\partial^2}{\partial t^2}g(x,t) + 3\frac{\partial}{\partial t}g(x,t) & =
0.01 \left( \frac{\partial^2}{\partial x^2}g(x,t) -
0.125\frac{\partial^4}{\partial x^4}g(x,t) \right)&& \\
&&\frac{\partial^2}{\partial t^2}g(x,t) + 3\frac{\partial}{\partial t}g(x,t) & =
0.01 \frac{\partial^2}{\partial x^2}g(x,t) && \\
&&\frac{\partial}{\partial t}g(x,t) & = 0.01 \frac{\partial^2}{\partial
x^2}g(x,t) && 
\end{flalign}
The corresponding polynomials are (i) $p_1(z) = z^2 + 3z, q_1(z) = 0.01(z^2 -
0.0125z^4)$, (ii) $p_2(z) = z^2 + 3z, q_2(z) = 0.01z^2$, and (iii) $p_3(z) = z,
q_3(z) = 0.01z^2$. Note that the last one is the diffusion equation.  The same
two sets of Fourier coefficients were used in first two equations. Note that the
others are conjugate of these to ensure real field. 
\begin{flalign}
&& a_1[0] = 0.3002;      \ &   a_2[0] = 0.2445;  \nonumber &&\\
&& a_1[1] = -0.0413 + j0.0216; \ & a_2[1] = -0.0357 + j0.0478;  \nonumber && \\
&& a_1[2] = 0.0871 + j0.0343; \ & a_2[2] = 0.0978 + j0.0729; \nonumber && \\
&& a_1[3] = -0.1679 - j0.0586;  \ &  a_2[3] = -0.1796 - j0.0756; &&
\end{flalign}
The Fourier coefficients used in the last equation are 
\begin{flalign}
&& a[0] &= 0.11 && \nonumber \\
&& a[1] &= 0.023 - j0.076 && \nonumber \\
&& a[2] &= 0.0669 + j0.0551 && \nonumber \\
&& a[3] &= 0.2 + j0.0821 && 
\end{flalign}

\begin{figure*}[!htb]
\centering
\includegraphics[width =6in]{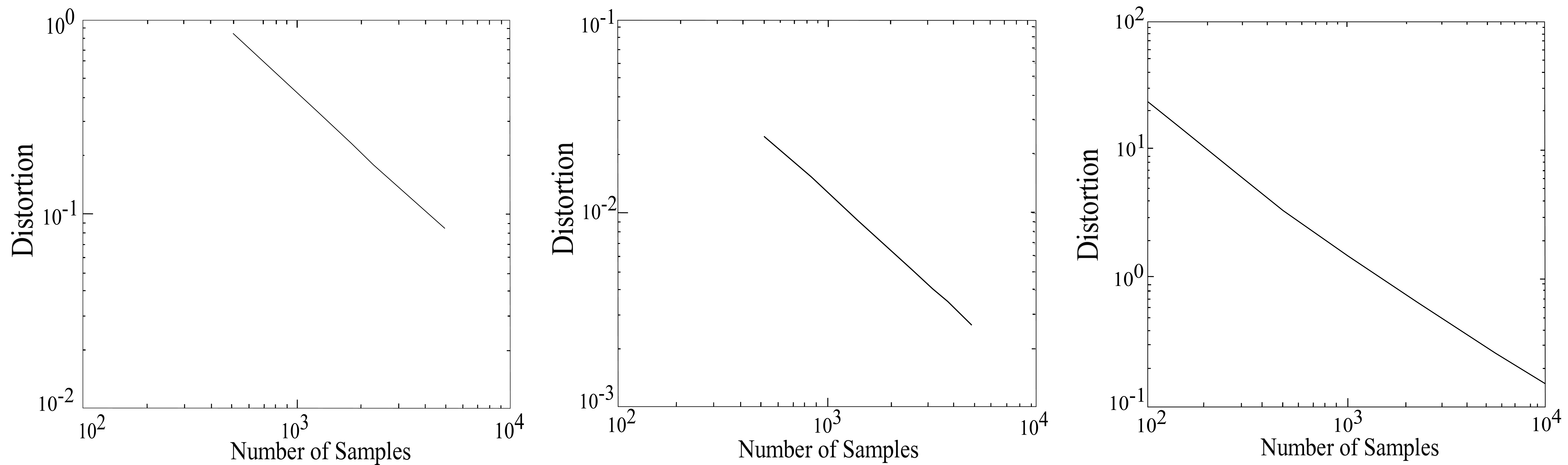}
\label{Fig :Simulations} \caption{ The three different figures show the
variation of error with the number of samples for different PDEs. The first one
is the corresponding to first PDE in equation \eqref{PDEs}, the middle one
corresponds to the second equation and the last one to the third equation in
\eqref{PDEs}, which is the standard diffusion equation. The variation in the
distortion is clearly of $O(1/n)$ as denoted by the slopes in the plots.
However, the error is slightly large because of the large condition number of
the matrix $Y_0$ giving issues regarding numerical stability.}
\end{figure*} 
The Figure \ref{Fig :Simulations} shows the mean square error of the estimate
for different PDEs. The plots are shown for the different PDEs as they have been
listed in the equation \ref{PDEs}. The slopes of the lines obtained are $-
1.0019, -1.0110$ and $-1.0086$ which confirms the $O(1/n)$ decrease.

\section{Conclusions}
\label{sec:conclusion}
The sampling of spatially bandlimited field evolving according to the constant
coefficient linear partial differential equation using a mobile sensor was
studied. The field was estimated using the noisy samples obtained at unknown
locations and time instants obtained from two independent and unknown renewal
processes and it was shown that the mean squared error between the estimated
field and the true field decreased as $O(1/n)$, where $n$ was the average
sampling density. The main idea that was leveraged was the fact that the
locations of the samples got closer to the ones corresponding to uniform
sampling as the sampling density increased and thus oversampling was used to
reduce the error.

\section*{Appendix A} % (fold)
\label{sec:appendix_a}
This appendix mainly deals with proving the bandlimitedness of the field for all
$t \geq 0$ given that the field was spatially bandlimited at $ t = 0$. Spatial
bandlimitedness, as defined previously, refers to the fact that the Fourier
series of the field over its spatial support has finite number of terms.The main
result shown is that for a field evolving according to the equation
\eqref{diff_eq_polynomial}, if it is known that the field and its $m - 1$,
temporal derivatives are spatially bandlimited at $t = 0$, then the field will
always remain bandlimited. Here $m$ is the degree of the polynomial $p$. A major
part of this proof will be similar to the approach in Section
\ref{sec:field_estimation_from_samples} and has been reproduced here for ease of
understanding. Since the field is assumed to have a finite support, the field
can be written as 
\begin{align}
g(x,t) = \sum_{k = \infty}^{\infty} a_k(t) \exp (j 2 \pi k x) ) \ ; \ a_k(t) =
\int_{-\infty}^{\infty} g(x,t) \exp(-j 2\pi k x) dx
\end{align}
Substituting this in the equation \eqref{diff_eq}, and using \eqref{polynomials}
along with the orthogonality property for the Fourier basis we can write,
\begin{align}
\sum_{i = 0}^m p_i \frac{\partial^i}{\partial t^i} \left( \sum_{k =
-\infty}^{\infty} a_k(t) \exp (j 2 \pi k x) \right) = \sum_{i = 0}^n q_i
\frac{\partial^i}{\partial x^i} \left(  \sum_{k = -\infty}^{\infty} a_k(t) \exp
(j 2 \pi k x) \right) \nonumber\\
\implies \sum_{k = -\infty}^{\infty} \left( \sum_{i = 0}^m p_i \frac{\partial^i
a_k(t)}{\partial t^i} \right) \exp (j 2 \pi k x) = \sum_{k = -\infty}^{\infty}
a_k(t) \left( \sum_{i = 0}^n q_i (j 2\pi k )^i \right) \exp (j 2 \pi k x)
\nonumber \\
\overset{\text{(a)}}{\implies} \sum_{k = -\infty}^{\infty} \left( \sum_{i = 0}^m
p_i \frac{\partial^i a_k(t)}{\partial t^i} \right) \exp (j 2 \pi k x) = \sum_{k
= -\infty}^{\infty} a_k(t) q(j 2 \pi k) \exp (j 2 \pi k x) \nonumber \\
\overset{\text{(b)}}{\implies}  \sum_{i = 0}^m p_i \frac{\partial^i
a_k(t)}{\partial t^i}  - q(j 2 \pi k) a_k(t) = 0 \  \ \forall k \in \mathbb{Z}
\end{align}
Essentially, each $a_k(t)$ evolves by an ordinary differential equation (ODE)
with constant coefficients. The solution of ODE with constant coefficients is
well known (via unilateral Laplace transform \cite{}). For each $k$, this leads
to the polynomial equation, 
\begin{align}\label{r_poly_appendix}
  \sum_{i = 0}^m p_i \frac{\partial^i Ae^{rt}}{\partial t^i}  - q(j 2 \pi k)
Ae^{rt} = 0 \nonumber \\ 
  \implies \left(\sum_{i = 0}^m p_i r^i  - q(j 2 \pi k) \right) Ae^{rt} = 0
\nonumber\\
  \implies p(r) - q(j 2 \pi k) = 0 
\end{align} 
The solution for $a_k(t)$ is a of the form $Ae^{rt}$, where $r$ is the root of
the above polynomial and $A$ is a constant independent of $t$. Let the roots of
the above polynomial be $r_1(k), r_2(k), \dots r_m(k)$. Note that the roots of
the polynomial are indexed by $k$ as well, implying there is a set of $m$ roots
for each value of $k$. Furthermore for simplicity of analysis, all of the roots
$r_1(k), r_2(k), \dots r_m(k)$ are considered to be distinct for a given $k$ as
in Section \ref{sec:field_estimation_from_samples}. However, it is possible that
$r_i(k_1) = r_j(k_2)$ for some $i, j, k_1 \neq k_2$.  In fact, the solution for
$a_k(t)$ can thus be written as a linear combination of these roots. Thus
$a_k(t) = \sum_{i = 1}^m a_{ki}(0)\exp(r_i(k) t)$. The coefficients have been
chosen to maintain consistency of represent as a function of time. Now, we know
that the field is bandlimited and so are its $m-1$ partial derivatives at $t =
0$. This can be written for $i = 0, 1, 2 \dots m-1$
\begin{align}
\frac{\partial^i a_k(t)}{\partial t^i}\bigg|_{t = 0} =
\begin{cases}
c_{ki}   & \quad \text{if }|k| \leq b\\
0  & \quad \text{otherwise} \\
\end{cases}
\end{align}
where $c_{ki}$'s are real constants. Since $a_k(t) = \sum_{i = 1}^m
a_{ki}(0)\exp(r_i(k) t)$, therefore,
\begin{align}\label{partial_derivs}
\frac{\partial^j a_k(t)}{\partial t^j} = \sum_{i = 1}^m
a_{ki}(0)r_i^j(k)\exp(r_i(k) t) \ \forall \ j \geq 0
\end{align}
Consider $k'$ in the range $|k| > b$. Then for $k'$ and $\forall \ i = 0, 1, 2
\dots m$, we have $\displaystyle \frac{\partial^i a_{k'}(t)}{\partial
t^i}\bigg|_{t = 0} = 0$. Then for $k'$ using equation \eqref{partial_derivs}, we
can combine all the equations for all $i$ and the resulting expression can be
written in matrix form as, 
\begin{align}
\begin{bmatrix}
1          & 1          & \dots  & 1\\
r_1(k')   & r_2(k')   & \dots  & r_m(k') \\
r_1^2(k') & r_2^2(k') & \ddots & r_m^2(k') \\
\vdots     & \vdots     & \ddots & \vdots \\
r_1^m(k') & r_2^m(k') & \dots  & r_m^m(k') \\
\end{bmatrix}
\begin{bmatrix}
a_{{k'}1}   \\
a_{{k'}2} \\
\vdots  \\
a_{{k'}m}  \\
\end{bmatrix} 
 = 0
\end{align}
Since the roots are assumed to be distinct, therefore the matrix on the left is
a Vandermonde matrix and is always invertible\cite{vandermode}. This means that
$a_{k'i} = 0 \ \forall i = 1,2,3 \dots m$ and all $|k'|> b$. This implies that
the field is bandlimited, i.e., $a_k(t) \equiv 0$ for all $|k| > b$.
%
%
% section appendix_a (end)

\section*{Appendix B} % (fold)
\label{sec:appendix_b}

This section primarily deals with establishing upper bound on the terms
$\displaystyle \mathbb{E}\left[ \frac{1}{M}\sum_{i = 1}^M \bigg|S_i -
\frac{i}{M}\bigg|^2 \right] $ and $ \displaystyle \mathbb{E}\left[
\frac{1}{M}\sum_{i = 1}^M \bigg|T_i - \frac{iT_0}{M}\bigg|^2 \right] $. The
renewal based sampling model for the spatial terms is the same that has been
considered in \cite{unknown_loc}. Moreover, the bound on the term $\displaystyle
\mathbb{E}\left[ \frac{1}{M}\sum_{i = 1}^M \bigg|S_i - \frac{i}{M}\bigg|^2
\right] $ has been elaborately derived there(\cite{unknown_loc}, Appendix A).
Using that we have the bound, 
\begin{align}
 \mathbb{E}\left[ \frac{1}{M}\sum_{i = 1}^M \bigg|S_i - \frac{i}{M}\bigg|^2
\right] \leq \frac{C_S}{n}
\end{align}
for some $C_S > 0$ and independent of $n$. The proof for the other term follows
in a similar manner. Define, 
\begin{align}
y_m := \mathbb{E}\left[ \left(N_1 - \frac{T_0}{M}\right)^2 \bigg| M = m \right]
\nonumber \\
z_m := \mathbb{E}\left[ \left(N_1 - \frac{T_0}{M}\right)\left(N_2 -
\frac{T_0}{M}\right) \bigg| M = m \right]
\end{align}
From equation $(31)$ in \cite{unknown_loc}, we have 
\begin{align}\label{AppendixB_1}
\mathbb{E}\left[ \frac{1}{M}\sum_{i = 1}^M \bigg|T_i - \frac{iT_0}{M}\bigg|^2
\bigg| M = m\right] = \frac{m+1}{2}y_m + \frac{m^2 -1}{3}z_m
\end{align}
Consider the expression,
\begin{align}
\mathbb{E}\big[(T_M - T_0)^2 | M = m\big] & = \mathbb{E}\left[ \bigg\{\sum_{i =
1}^M \left(N_i - \frac{T_0}{M}\right)\bigg\}^2 \bigg| M = m\right]\nonumber \\
& = my_m + m(m-1)z_m 
\end{align}
where the second step is obtained from evaluating and rearranging the expression
along the exchangeability of $N_i$'s. Since we know that, $T_M \leq T_0$ and
$T_{M+1} > T_0$, define 
\begin{align}
J_M = T_0 - T_M \implies J_M < T_{M+1} -T_M \leq \frac{\mu}{n}
\end{align}
Also $\displaystyle J_M^2 = (T_M - T_0)^2 \implies \mathbb{E}\big[(T_M - T_0)^2
| M = m\big] = \mathbb{E}[J_M^2 | M =m] \leq \frac{\mu^2}{n^2} $.  Using the
above two results, we can conclude that, $\displaystyle my_m + m(m-1)z_m =
\mathbb{E}[J_M^2 | M =m]$. Combining this with \eqref{AppendixB_1}, we can
write,
\begin{align}
\mathbb{E}\left[ \frac{1}{M}\sum_{i = 1}^M \bigg|T_i - \frac{iT_0}{M}\bigg|^2
\bigg| M = m\right] & = \frac{m+1}{2}y_m + \frac{m^2 -1}{3m(m-1)}\left(-my_m +
\mathbb{E}[J_M^2 | M =m]\right) \nonumber \\
& = \frac{m+1}{2}y_m + \frac{m + 1}{3m}\mathbb{E}[J_M^2 | M =m] \nonumber \\
& \leq \frac{m+1}{2}y_m + \frac{2}{3} \frac{\mu^2}{n^2}
\end{align}
This is exactly the same result as obtained for $\displaystyle \mathbb{E}\left[
\frac{1}{M}\sum_{i = 1}^M \bigg|S_i - \frac{i}{M}\bigg|^2 \right]$ in
\cite{unknown_loc}. Using the same steps for the expression in $S_i$, we can
conclude,
\begin{align}
\mathbb{E}\left[ \frac{1}{M}\sum_{i = 1}^M \bigg|T_i - \frac{iT_0}{M}\bigg|^2
\bigg| M = m\right] \leq \frac{C_T}{n}
\end{align}
An important thing to note here is that even if $J_M \leq \frac{K\mu}{n}$ for
some positive constant $K$, the result will hold. Interestingly, $K$, can be
$O(\sqrt{n})$, and still the result will hold. The idea is that the difference
between $T_0$ and $T_M$ should of $O(1/\sqrt{n})$, that is $T_0$ cannot be
simply any number larger than $T_M$. It has to be a reasonably accurate
estimation of the time taken. It is important to know this to help decide the
construction of $Y_0$ because the entire idea is based on the assumption that
the samples are ``near''the grid points. But to determine the grid points, we
must have the knowledge of the support of the function which has to be finite.
% section appendix_b (end)

\end{document}